%% file: ITCS_Submission.tex
\documentclass{article}

\usepackage{amsmath,ulem}
\usepackage{amssymb}
\usepackage{amsthm}
\usepackage[utf8]{inputenc} 
\usepackage[T1]{fontenc}    
\usepackage{hyperref}       
\usepackage{url}            
\usepackage{booktabs}       
\usepackage{amsfonts}       
\usepackage{nicefrac}       
\usepackage{microtype}      
\usepackage{xcolor}         

\newtheorem{theorem}{Theorem}

\newtheorem{definition}{Definition}
\newtheorem{example}{Example}

\usepackage[colorinlistoftodos,backgroundcolor=lightgray!50,linecolor=black,textsize=scriptsize]{todonotes}

\newcommand{\blank}{`\mathrm{blank}\text{'}}

\newcommand{\G}{\ensuremath{\mathcal{G}}}
\newcommand{\A}{\ensuremath{\mathcal{A}}}
\renewcommand{\H}{\mathcal{H}}
\newcommand{\V}{\mathcal{V}}

\newcommand{\E}{\mathcal{E}}

\newcommand{\C}{\ensuremath{\mathcal{C}}}
\newcommand{\Z}{\ensuremath{\mathcal{Z}}}

\newcommand{\F}{\mathcal{F}}
\newcommand{\Q}{\mathbb{Q}}
\newcommand{\M}{\mathcal{M}}

\renewcommand{\S}{\mathcal S}
\newcommand{\N}{\mathbb{N}}

\input{mathdefs.tex}

\newcommand{\PGedit}[1]{\textcolor{black}{#1}}

\newcommand\independent{\protect\mathpalette{\protect\independenT}{\perp}}
\def\independenT#1#2{\mathrel{\rlap{$#1#2$}\mkern2mu{#1#2}}}

\title{A Minimal Intervention Definition of Reverse Engineering a Neural Circuit}

\author{%
  Keerthana Gurushankar$^\dagger$ and Pulkit Grover$^*$ \\
  $\dagger$ Department of Mathematics\\ $*$Department of Electrical \& Computer Engineering, and Neuroscience Institute\\
  Carnegie Mellon University\\
  Pittsburgh, PA 15213 \\
  \texttt{kgurusha@andrew.cmu.edu, pulkit@cmu.edu } \\
}

\begin{document}

\maketitle

\begin{abstract}
In neuroscience, researchers have developed informal notions of what it means to reverse engineer a system, e.g., being able to model or simulate a system in some sense. A recent influential paper of Jonas and Kording, that examines a microprocessor using techniques from neuroscience, suggests that common  techniques to understand neural systems are inadequate. Part of the difficulty, as a previous work of Lazebnik noted, lies in lack of formal language. Motivated by these papers, we provide a theoretical framework for \textit{defining} reverse engineering of computational systems, motivated by the neuroscience context. Of specific interest to us are recent works where, increasingly, interventions are being made to alter the function of the neural circuitry to both understand the system and treat disorders. Starting from Lazebnik's viewpoint that understanding a system means you can ``fix it'', and motivated by use-cases in neuroscience, we propose the following requirement on reverse engineering: once an agent claims to have reverse-engineered a neural circuit, they subsequently need to be able to: (a) provide a minimal set of interventions to change the input/output (I/O) behavior of the circuit to a desired behavior; (b) arrive at this minimal set of interventions while operating under bounded rationality constraints (e.g., limited memory) to rule out brute-force approaches. Under certain assumptions on the model, we show that this reverse engineering goal falls within the class of undecidable problems, by connecting our problem with Rice's theorem. Next, we examine some canonical computational systems and reverse engineering goals (as specified by desired I/O behaviors) where reverse engineering can indeed be performed. Finally, using an exemplar network, the ``reward network'' in the brain, we summarize the state of current neuroscientific understanding, and discuss how computer-science and information-theoretic concepts can inform goals of future neuroscience studies.

\end{abstract}


\vspace{-0.1in}
\section{Overview}\label{sec:intro}
\vspace{-0.1in}



Two works in this century point to the need for formal definitions and rigor in understanding neural computation. The essay of Lazebnik~\cite{lazebnik2002can}, provocatively titled ``Can a biologist fix a radio,'' emphasizes on the need for a formal language to describe elements and questions within biology so that there is reduced ambiguity or vagueness, and clear (falsifiable) predictions are made. This need is becoming increasingly evident in attempts to reverse engineer the brain. While neural recording and stimulation technology is advancing rapidly\footnote{Neural recordings are undergoing their own version of ``Moore's law'': the number of neurons being recorded simultaneously is increasing exponentially~\cite{stevenson2011advances}.}, and techniques for analyzing data with statistical guarantees have also expanded rapidly, the techniques do not provide satisfying answers for understanding the system~\cite{jonaskording,gao2015simplicity}. This is most evident in the strikingly detailed work of Jonas and Kording~\cite{jonaskording}\footnote{Titled ``\textit{Could a neuroscientist understand a microprocessor?}'',~\cite{jonaskording} follows in the footsteps of Lazebnik's, but also tests popular  techniques from  computational neuroscience. See also the \textit{Mus Silicium} project~\cite{hopfield2001moment}.}, which use an early but sophisticated microprocessor, MOS 6502, instead of Lazebnik's radio. They examine this microprocessor under 3 different ``behaviors'' (corresponding to 3 different computer games, namely, Donkey Kong, Space Invaders, and Pitfall), and conclude that ``... current analytic approaches in neuroscience may fall short of producing meaningful understanding of neural systems, regardless of the amount of data''. The work also underscored the need for rigorous testing of tools on simulated data prior to application on real data for obtaining inferences. Because they focus on  concrete implementations and a fully specified and simple system, they conclude that they should obtain an understanding that ``guides us towards the descriptions'' commonly used in computer architecture (e.g., an Arithmetic Logical Unit consisting of simple units such as adders, a memory). Subjective definitions of reverse engineering have been explored elsewhere as well (e.g.~\cite{BinYuPaper,doshi2017towards}). 

Inspired by~\cite{lazebnik2002can,jonaskording}, we ask the normative question: what end-goal for reverse engineering \textit{should} the neuroscientists aim for? Our main intellectual contribution in this context can be summarized in two pieces: a) Viewing reverse engineering as \textit{faithful summarization}, i.e., one needs to be represent the computation not just faithfully but also economically; and b) Specifying what may constitute faithful representation of a computation in the context of neuroscience. Specifically, we take an \textit{minimal-interventional} view of faithful representation, as explained below.

\textit{Reverse engineering is faithful summarization}: The act of modeling/abstracting itself is compression, as good models tend to preserve the essence of the phenomenon/aspect of interest, discarding the rest~\cite{galloway2017compression}. This is also reflected in neuroscience-related works~\cite{craver2006mechanistic}. Literature in Algorithmic Information Theory, which uses Kolmogrov complexity (minimal length of a code to compute a function) to quantify degree of compression, has also been connected to understanding~\cite{chaitin2005epistemology}. E.g., a reverse engineering agent (human or artificial) should be able to compress the description of the computational system in a few \textit{bits}. The \textit{degree} to which the description can be compressed, while still maintaining a faithful representation, quantifies the \textit{level or degree} of understanding (i.e., reverse engineering). This compression rules out, for instance, brute-force approaches that store a simulation of the entire computational system as reverse engineering (discussed further in Section~\ref{sec:background}).

\textit{What constitutes faithful representation}: How do we quantify faithfulness of a representation? We believe it is important to not just preserve the input/output (I/O) relationship, but also preserve \textit{how} the function is computed, summarizing relevant information from the structure and architecture of the network and the function computed at each of the nodes (e.g., the structure of the Fast Fourier transform, FFT Butterfly network, considered in Section~\ref{sec:examples}, is integral to how the FFT is often implemented). In other words, preserving only the I/O relationship misses the point of \textit{how} the computation is carried out (it preserves, exclusively, \textit{what} function is implemented, but not how). Motivated by operational goals of understanding  implementation as a way of understanding how the computation is performed, we impose an \textit{interventional} requirement on faithful representations, namely, that a representation is faithful if it enables predicting \textit{minimal interventions that change the I/O behavior of the system from the existing behavior to another desired behavior}. Our emphasis on minimal interventions is because we want to rule out approaches that change the entire system to another system (i.e., those that only rely on the I/O relationship and not the structure/implementation, e.g., an approach that replaces the entire system with one that has a desirable I/O behavior might not be a minimal intervention).

Tying the two aspects above together, we arrive at our definition of reverse engineering (more intuition in Section~\ref{sec:background}, formally stated in Section~\ref{sec:definition}). Informally, one must be able to summarize the description using just a few bits, and this description should suffice for minimal interventions to change the I/O relationship to a desired one. 

Our interventional definition is not without precedence. Indeed, a classical (if informal) view of understanding a system requires that one must be able to \textit{break it into pieces and put it back together}, or, in Lazebnik's words~\cite{lazebnik2002can}, ``fix'' it.  Some existing approaches in explainable/interpretable machine-learning also use  interventions to understand the system, e.g., influence of features on the output~\cite{anupamQuantitative}. This might offer an achievability of reverse engineering, but our work is distinct in that it attempts to \textit{define} explainability in an interventional sense. Here,  our goal is one of editing the \textit{network} (and not just the features) to demonstrate understanding. Interventionist accounts of explanations have been discussed in philosophy of science. Woodward~\cite{woodward2011scientific} argues in support of explanations that describe not only the I/O behavior of the system, but also the behavior after interventions. In the context of neuroscience,  Craver~\cite{craver2006mechanistic}, among others, separates ``explanatory models'' from ``phenomenally adequate''. Whereas phenomenally adequate models might only describe or summarize the phenomenon of interest, explanatory models should also allow a degree of control and manipulation. 

These views are well aligned with ours. Additionally, our work (specifically, the minimal interventions aspect) is motivated by advances in neural engineering and clinical efforts in treating disorders. Recent efforts have succeeded in engineering systems (e.g. neural dust, nanoparticles, injectable electronics~\cite{neuraldust,trevathan2019injectable,xu2020remote}) that can be implanted with minimal tissue damage, and are being tested in animal experiments (even noninvasive techniques are increasing in their precision~\cite{JNEPaper}). Recent clinical efforts in humans have involved chronic (i.e., long-term) implantation of electrodes for treating depression~\cite{bewernick2010nucleus}, obsessive-compulsive disorder (OCD)~\cite{blomstedt2013deep}, addiction~\cite{pierce2013deep}, obesity~\cite{oh2009deep}, etc., which are all disorders of the reward network discussed in Section~\ref{sec:reward}. One clinical end-goal is to manipulate this circuit with minimal interventions. Where do we place and when do we activate the neural implants, and what is the effect they should produce? Our work casts this question in a simplified and abstract model.

{In explainable AI literature, there is an acknowledgment that being able to propose interventions is a way of demonstrating understanding of a decision-making system~\cite{doshi2017towards,sokol2020explainability,molnar2020interpretable}, although much of this body of work is focused on interventions on the feature space~\cite{Sundararajan2017Axiomatic,Dabkowski2017Real,Bhatt2020Evaluating} or individual data points~\cite{Koh2017Understanding,Kim2016Examples}, rather than inside the computational network. Rob Kass, a noted neuroscientist-statistician, notes in his Fisher lecture~\cite{kassLecture}, using the example of the brain's reward circuitry~\cite{russo2013brain}, that the goal of tools that describe information flow can be to obtain interventions (e.g. using neurostimulation) on the system. He suggests that understanding information flow can help identify optimized interventions to treat disorders such as anxiety and addiction, both related to the reward network~\cite{russo2013brain}. In AI, it is often not required for explanations to be at a physical implementation level. In neuroscience, as noted here, explanations tied to the implementation can help with interventions for treating disorders (specially with recent advances in neuroengineering). }

 \textit{What this work accomplishes}. The main contribution of this paper is 3-fold, (\textit{i}) \textit{the reverse-engineering definition itself}, stated formally in Section~\ref{sec:definition}. (\textit{ii}) \textit{An undecidability result}: In the spirit of formal treatments, even under optimistic assumptions on what can be learned about the system through observations and interventions, we obtain a hardness/impossibility result, showing that a sub-class of the general set of reverse engineering problems is undecidable, i.e., no agent which is itself a Turing machine can provide a desirable reverse engineering solution for arbitrarily chosen problems for our minimal-interventions definition. This result is obtained by connecting Rice's theorem in theoretical computer science~\cite{rice1953classes} with our reverse engineering problem, and is the first connection drawn between neuroscience and Rice's theorem. {Further, to illustrate how this result about the undecidability of reverse engineering is not merely an artifact of our chosen definitions, we also include alternative plausible definitions of reverse engineering, and proofs of their undecidability in Appendix \ref{sec:alternative_definitions}}; (\textit{iii}) \textit{Examples}: In Section~\ref{sec:examples}, we illustrate that this goal is attainable in interesting (if toy) cases, by using examples of simple computational systems, including a toy network inspired by the reward network in the brain, and describing their reverse engineering solutions. {Additionally, in Section~\ref{sec:reward}, we discuss \textit{an exemplar neural circuit: the reward network}. We overview the state of understanding of this exemplar circuit and discuss what it may lack from our reverse engineering perspective}. We conclude with a discussion in Section~\ref{sec:discussion}, including limitations of our work. 

 {\textit{Place within (and outside) TCS's scope and literature}: In Section~\ref{sec:background}, we provide a more detailed literature review to help position the main contribution of our work in the neuroscience context (i.e., \textit{outside} CS-theoretic context). Within the theoretical computer science context, we view our main contribution to be the definitions and a connection  with models used in neuroscience (see, e.g. models in~\cite{venkatesh2020information,Friston2011Functional}, etc.). This allows us to formally examine neuroscience questions using CS-theoretic techniques, connecting the context of neuroscience with techniques from CS-theory (in particular  Rice's theorem). The specific undecidability results simply fall out of making this formal connection (see also Appendix~\ref{sec:alternative_definitions}). More broadly, modifications on our approach and models can pave the way to more formal treatment of neuroscience problems from a CS theoretic lens, including complexity-theoretic and algorithmic advances on problems of reverse engineering. 

\section{Background and related  neuroscience work}\label{sec:background}

Explicitly or not, the question posed here connects with all works in neuroscience. Thus, rather than task ourselves with the infeasible goal of a thorough survey, we strive to illustrate the evolution of the relevant neuroscience discussion.

Perhaps the simplest reverse-engineering of a computational system is being able to ``simulate'' the I/O behavior of the system (see Introduction of~\cite{jonaskording}). E.g., cochlear and retinal prostheses attempt to replace a (nonfunctional) neural system with a desirable system with ``healthy'' I/O behavior (see also~\cite{horiuchi1993analog,berger2011cortical} for examples of such attempts for sensory processing and memory, respectively). This ``black-box'' way of thinking may suffice for understanding what is being computed\footnote{However, we acknowledge that I/O behavior can also have more or less understandable descriptions, e.g. machine-learning models of different complexity approximating the same I/O relationship. Thus, while it is not a focus of this work, a black-box way of describing I/O relationships has more nuance to it than is discussed here.}, but not \textit{how}. To describe \textit{how} a computation is being performed, one might seek to describe the input-output behavior of individual elements of computation (which could be as fine-grained as compartments of a single neuron, or a neuron itself, or a collection of neurons). There is a compelling argument that even this component-level simulation is insufficient. E.g., Gao and Ganguli~\cite{gao2015simplicity}, in their work on required minimal measurements in neuroscience, note that while we can completely simulate  \textbf{artificial} neural networks (ANNs), most machine-learning researchers would readily accept that we do not understand them. This led Gao and Ganguli to ask: ``$\ldots$ can we say something about the behavior of deep or recurrent ANNs without actually simulating them in detail?'' (see related field of ``explainable machine-learning''~\cite{BinYuPaper,dovsilovic2018explainable}). That is, a component-level understanding can miss an understanding at an intuitive level. 

To state what a more comprehensive understanding of a computational system could look like, inspired by the visual system, cognitive scientist David Marr proposed ``3 levels of analysis''~\cite{marr1982}: computational, algorithmic, and implementation. At the lowermost, \textit{implementation level}, is the question of how a computation is implemented in its hardware. Above that, at the algorithmic level, the question, stated informally by Marr, is what algorithm is being implemented, e.g., how it represents information and modulates these representations. Finally, at the highest level is the problem being solved itself. We refer the reader to~\cite{peebles2015thirty} for some of the recent discussions on Marr's levels. 
 Gao and Ganguli write in agreement, with subtle differences: ``understanding will be found when we have the ability to develop simple coarse-grained models, or better yet a hierarchy of models, at varying levels of biophysical detail, all capable of predicting salient aspects of behavior at varying levels of resolution''.\footnote{Thereon, Gao and Ganguli connect the problem of evaluating the minimum number of required measurements as a metric for understanding the system. This view is inspired by the success of modern machine-learning approaches, but might find disagreement from Chomsky~\cite{chomsky}.} While influential and useful, Marr's and Gao/Ganguli's descriptions are too vague to quantify reverse engineering in a formal sense. 



An exciting alternative approach was recently proposed by Lansdell and Kording~\cite{KordingRulesOfLearning}. Motivated by lack of satisfactory understanding of ANNs, their approach is to change the goals. 
They ask the question: can we learn the rules of learning, and could that be a pathway to reverse engineering cognition? This is an interesting approach worthy of further examination, but is not directly connected with this current work.

As discussed in Section~\ref{sec:intro}, complementary to these lines of thought, we take a fundamentally interventional view of reverse engineering.  We also strive, in the established information-theoretic and theoretical computer science traditions, to state the problem formally, and then observe fundamental limits and achievabilities. This goal is challenging, to say the least, but efforts in this direction are needed to ground the questions in neuroscience concretely.

\section{Our minimal intervention definition of reverse engineering}\label{sec:definition}
\textbf{Overview of our definition and rationale for our choices}: We allow the agent performing the reverse engineering to specify several classes of desirable I/O relationships. To constrain the agent from using brute-force approaches, if the agent claims to have successfully reverse engineered the system, it must be able to \textit{produce a Turing machine} that requires only a limited number of bits to describe. This Turing machine should be able to take a class of desirable I/O relationships as input, and provide as output a set of interventions that change the I/O relationship to one of the desirable ones within this class. The rationale for the requirement on the agent to provide a Turing machine is that it is a complete description of the summarization. An informal ``compression'' to a certain number of bits could hide the cost of encoding and decoding, or of some of the instructions in execution of the algorithm.  The rationale of allowing any one of a class of I/O behaviors as an acceptable solution is that it allows for approximate solutions or choosing one among solutions that are (nonuniquely) optimal according to some criteria (e.g. in the reward circuitry which drives addiction, discussed in Section~\ref{sec:examples}, any I/O behavior that eliminates the reward of an addictive stimulus might suffice). 

In addition, we allow the Turing machine to have a few accesses to the computational system where it can perform interventions and observe the changed I/O relationship. While this still disallows brute-force approaches, it enables lowering the bar on what is required for reverse engineering.

These definitions are there to lay down a formal framework in which we can obtain results. They can easily be modified. In arriving at this reverse engineering solution (i.e., in generating the Turing Machine), we allow the agent to access the ``source code'' of the computaional system $\C$. This might appear to be an optimistic assumption (indeed it is so) as it might require noiseless measurements everywhere, and possibly causal interventions, which current neuroengineering techniques are very far from. The definition can readily be modified to include access to limited noisy observations, which will only make the reverse engineering harder. Note that with the ``Moore's law of neural recording,'' it is conceivable that each node and edge can indeed be recorded in distant (or nearby) future~\cite{stevenson2011advances}.  As another example, while we assume, for simplicity, that communication happens at discrete time-steps, this assumption can be relaxed for some of our results, e.g., our undecidability result in Section~\ref{sec:impossible} because it only makes the reverse engineering problem harder). Similarly, equipping the system with an additional external memory (e.g., the setup in~\cite{NeuralTuringMachines}) also makes the reverse engineering problem harder.

\subsection{System model}\label{sec:sysmodel}


\begin{definition}[Computational System and Computation] \label{def:comp-sys}
	A \emph{computational system} $\C$ is defined on a finite directed graph $\G(\V,\E)$, which is a collection of nodes $\V$ connected using directed edges $\E$. The computation uses symbols in a set $\S\subseteq \mathbb{R}$ ($\S$ is called the ``alphabet'' of $\C$), where $0\in\S$. Each node $\V$ stores a value in $\S$ (initialized to any fixed $s\in\S$). The computational input is a finite-length string of symbols in $\S$. The computation starts at time $t=0$ and happens over discrete time steps. At each time step, the $i$-th node, for any $i$, computes a function on a total of $n_i$ symbols, which includes (i) symbols stored in each node from which it has incoming edges (called ``transmissions received from'' the nodes they are stored in), (ii) the symbol stored in the node itself, and (iii) at most one symbol from the computational input. The node output at any time step, also a symbol in $\S$, replaces the stored value. That is, the $i$-th node computes a function $\S^{n_i}\to \S$, mapping the $n_i$ symbols from the previous time instant (including nodes with incoming edges, the locally stored value, and the computation input) to update its stored value. The stored values across all nodes collectively form the ``state'' of the system at each time instant. A set of nodes are designated as the output nodes, and their first nonzero transmissions are together called the output of the computation. 
\end{definition}
%
This description of $\C$, with  $\G$ and the functions computed at the nodes, is called the ``source code'' of $\C$. This definition is inspired by similar definitions in information theory and theory of computation~\cite{Ahlswede2000Network,Thompson1980Complexity}, including a recent use in neuroscience~\cite{venkatesh2020information}.





\begin{definition}[Input/Output (I/O) relationship of $\mathcal{C}$]
The input-output relationship (I/O relationship) of $\mathcal{C}$ is the mapping from the inputs to $\mathcal{C}$ to the outputs of $\mathcal{C}$. 
\end{definition}

\begin{definition}[Interventions  on $\mathcal{C}$]
A single intervention on $\mathcal{C}$ modifies the function being computed at exactly one of the nodes in $\mathcal{C}$ at exactly one time instant. 
\end{definition}
An intervention would commonly change the I/O relationship of $\mathcal{C}$. 



\subsection{Definition of reverse engineering}
As discussed, our definition in essence is about making the system do what you want it to do. One way to view this, consistent with ``fixing'' the system, is by modifying the system $\C$, we should be able to get the input-output relationship we desire. 

Some notation: we will use $\H=\{\F_p\}_{p\in\mathcal{P}}$ (for a countable index set $\mathcal{P}$) to denote a collection of sets $\F_p$ where each $\F_p$ is a set of I/O relationships obtainable by multiple interventions on $\C$. 
Intuitively, each element $\F\in\H$ represents a set of I/O relationships that are ``equivalent'' from the perspective of the end-goal\footnote{Note that, because $\F_p$ need not be disjoint sets, our definition allows two I/O relationships to be equivalent w.r.t. one $\F_p$ but not w.r.t. another $\F_p$.} of  interventions on $\C$. For instance, they could all approximate a desirable I/O relationship. As an illustration for the reward network, say $\H=\{\F_1,\F_2\}$, where $\F_1$ is the set of I/O relationships corresponding to unhealthy addiction, whereas $\F_2$ might represent I/O relationships corresponding to healthy motivation. 

To perform these interventions, we now define an agent $\A$, whose goal is to generate a Turing machine that takes as input an index $p$, and provides as output the necessary interventions on $\C$ to attain a desirable I/O relationship $g\in\F_p$.

\begin{definition}[Reverse Engineering Agent $\A$ and $M$-bit summarization]
An agent $\A$ takes as input the source-code of $\C$ and $\H$, a collection of sets of I/O relationships, and outputs a Turing-machine $TM_{\C,M,\H,Q}$ which is described using no more than $M$-bits. We refer to $TM_{\C,M,\H,Q}$ as an $M$-bit summarization of $\C$. $TM_{\C,M,\H,Q}$ takes as input $p\in\mathcal{P}$. Additionally, $TM_{\C,M,\H,Q}$ also has access to an oracle to which it can input up to $Q$ different sets of multiple interventions on $\C$, and $p'\in\mathcal{P}$. For each set of interventions, the oracle returns back whether the resulting I/O relationship for a set of multiple interventions lies in $\F_{p'}\in\H$. For any input $p\in\mathcal{P}$, $TM_{\C,M,\H,Q}$ outputs a set of interventions $\Z$ on $\C$. It can also declare ``no solution''.
\end{definition}

Inspired by similar bounded-rationality approaches in economics and game theory~\cite{simon1990bounded,papadimitriou1994complexity,sims2005rational}, the $M$-bit summarization can enforce a constraint on $\A$ that disallows brute-force approaches, e.g., where $\A$ simply stores the changes in I/O relationships for all possible sets of interventions, and for a given reverse-engineering goal, simply retrieves the solution from the storage. We now arrive at our definition of reverse engineering.


\begin{definition}[$(\H,L,M,Q)$-Reverse Engineering]
\label{def:weakRE}
Consider a computational system $\mathcal{C}$ with an I/O relationship described by $f(\cdot)$. Let $\mathcal{A}$ be an agent that is claimed to have $(\H,L,M,Q)$ reverse engineered $\mathcal{C}$. Then, for a given $p\in\mathcal{P}$  that is input to the Turing machine $TM_{\C,M,\H,Q}$ (which was generated by $\A$), the output should be a set of interventions $\Z$ of the smallest cardinality (if $|\Z|\leq L$) that change the I/O relationship from $f(\cdot)$ to any $g\in \F_p$ (but not necessarily for all $g\in\F_p$). If no such $g\in \F_p$ exists,  then $TM_{\C,M,\H,Q}$ should declare ``no solution'', i.e., no such set of ($L$ or fewer) interventions exists.

\end{definition}







\section{Undecidability of some reverse engineering problems}\label{sec:impossible}


Reverse engineering is not undecidable for every class of $\C$'s, the class has to be rich enough. Below, we first prove a result on how rich the class needs to be for it to be Turing-equivalent. Following this result, we use Rice's theorem~\cite{rice1953classes,RicesThm} to make a formal connection with reverse engineering, proving in Theorem~\ref{thm:undecidability} that for set of $\C$'s that use an $\S$ of infinite cardinality, and computable functions at each node, the reverse engineering in Definition~\ref{def:weakRE} is undecidable for nontrivial $\H$'s, i.e., no agent $\A$  that is itself a Turing Machine can provides a reverse engineering solution for every $\C$ in this class for any $L\geq 0$, any $M$ (including $M=\infty$), and $Q=0$. Our undecidability result (Theorem~\ref{thm:undecidability}, which uses Theorem~\ref{thm:equivalence}(2) that is proven for a more limited set of $\C$'s) is for a more restricted class (specifically, the $\C$'s that can simulate ``$\sigma$-processor nets'' of~\cite{sontag95}) of computational systems than allowed in Definition~\ref{def:comp-sys}. Hence, reverse engineering the broader class (for which Theorem~\ref{thm:undecidability} is  stated) would only be harder (and hence is also undecidable). 

\begin{theorem}\label{thm:equivalence}
		(1) If {$|\S|$ is finite}, then the class of $\C$'s in Def.~\ref{def:comp-sys} is equivalent to deterministic finite-state automaton (DFA).\\
		(2) If {$|\S|$ is countably infinite} (e.g. $\Q$) and all nodes compute computable functions, then the class of $\C$'s is Turing equivalent.\\
		(3) If the function at any node is uncomputable, then the class of $\C$'s is super-Turing.
\end{theorem}

\begin{proof}
	Proof overview of (1): We construct a $\C$ (with {finite $\mathcal{S}$}) that simulates a given DFA (full description in the Appendix) as follows: the nodes and edges correspond to the states and transition edges of the DFA. We include an additional output node with incoming edges from all other nodes. When the DFA is in some state $q$, the corresponding node (the ``active'' node) of $\C$ is set to the computational input $d_t$ just received. All remaining nodes store a $\blank$ value. Suppose the DFA transitions to state $q'$ upon receiving $d_t$, then the corresponding node $q'$ sets to the next computational input $d_{t+1}$, becoming the active node in the next time-step. All other nodes are set to $\blank$. Finally, after receiving the full input string, the output node sets to $1$ or $0$ based on whether the last active node of $\C$ corresponded to an accepting/rejecting DFA state.
	
	A DFA can also simulate a computational system $\G(\V,\E)$ with finite $\S$ and $n=|\V|$ nodes as follows: the DFA  (i) has state space $Q=\S^\V$; (ii) has alphabet $\Sigma=\S^\V$; and (iii) starts in the state of each node holding the initial value. The transition function $\delta:Q\times\Sigma\to Q$ is defined as
	    \[
	        \delta((s_1,\dots, s_n), (d_1, \dots, d_n)) = \left(f_i(\left(s_j\right)_{j\in N(i)}, d_j) \right)_{i\in \V}
	    \]
	    and accepting states $\{(s_v)_{v\in \V} : s_o =1\}$ where $s_o$ is the output node of $\C$. The DFA accepts an input string iff the output node of $\C$ would set to $1$ upon receiving the string. 

	Proof of (2): To show the Turing completeness of the class of $\C$'s, we show Turing completeness of a smaller class, namely the set of $\C$'s that can simulate ``$\sigma$-processor nets'', defined in Siegelmann \& Sontag~\cite{sontag95}, which are a model for artificial neural nets operating on rationals using sigmoidal functions $\sigma(\cdot{})$ (see~\cite{sontag95} for details). Siegelmann \& Sontag showed that $\sigma$-processor nets are Turing complete~\cite{sontag95}. Thus it is sufficient to show that the class of $\C$'s can simulate $\sigma$-processor nets, which follows from the following: a $\sigma$-processor net $\mathcal{N}_\sigma$, upon receiving data and validation bits $d, v$, computes $x \mapsto \hat{\sigma}(Ax+db+vb'+c)$ for some matrix $A\in\Q^{N\times N}$ and vectors $b, b', c\in\Q^N$. For each  $\mathcal{N}_\sigma$, we make a computational system on the following directed graph: $N$ nodes, one for each state of $\mathcal{N}_\sigma$, and all edges, with the function computed at node $i$ being
	\[
	    f_i(x_1, \dots, x_n, d, v) = \sigma\left(\sum_{j=1}^N A_{ij}x_j +db_i+vb'_i+c_i\right).
    \]
    
	Proof of (3): Consider a computational system with infinite $\S$, consisting of a single node outputting an uncomputable function of the input $d\in \S$. Since the function is uncomputable, there is trivially no Turing machine capable of simulating it. 
\end{proof}

\begin{definition}[Nontrivial set of languages]
    The set of inputs accepted by a Turing machine is called its \textit{language}. An input string is accepted by a Turing machine if the computation terminates in its accept state (see, e.g.~\cite[Ch 3]{sipser13} for definition). 
    Alternatively, the computation could loop forever or terminate in a reject state.
 A Turing-Recognizable language is one for which there exists a Turing Machine that accepts only the strings in the language, and either rejects or does not halt at other strings. A set of languages $S$ is  \textit{nontrivial} if there exists a Turing-Recognizable language that it contains, and a different  Turing-Recognizable language that it does \textbf{not} contain. 
\end{definition}

{Any I/O relationship for $\C$ can be reduced to a decision problem/language (i.e. a mapping from finite string input to binary ``accept/reject" output) by designating one of its possible outputs as ``reject", and accepting strings with any other output. Thus, an I/O relationship for $\C$ can be viewed as a language of $\C$. Thus, our definition of I/O relationship sets $\F_p$ naturally extends to \textit{nontrivial} $\F_p$'s.} We now state Rice's theorem (Theorem~\ref{thm:rice}), which provides an undecidability result that we rely on to derive our undecidability result (Theorem~\ref{thm:undecidability}) by connecting our class of $\C$'s with Rice's theorem. While originally proven by Rice in~\cite{rice1953classes}, for simplicity, we use the statement of Rice's theorem from~\cite{RicesThm}. 

\begin{theorem}[Rice's theorem~\cite{rice1953classes,RicesThm}]\label{thm:rice}
 Let $S$ be a nontrivial set of languages. It is undecidable whether the language recognized by an arbitrary Turing machine lies in $S$. 
\end{theorem}
\begin{theorem}\label{thm:undecidability}
	{For an $\H$ containing a nontrivial $\F_p$}, for any $L\geq 0$, $M= \infty$, and $Q=0$, there is no Turing machine $\A$ which can accept as input, an arbitrary computational system $\C$ with infinite set $\S$ and computable functions evaluated at nodes, and output $TM_{\C,M,\H,Q}$ that satisfies the reverse engineering properties in Definition~\ref{def:weakRE}.
\end{theorem}

\begin{proof}
    Assuming there were such a Turing machine $\A$, we construct a Turing machine $\M$ (that will solve Rice's problem) as follows:  accept input string $s$ encoding Turing machine corresponding to $\C$ (Theorem~\ref{thm:equivalence} states that, with infinite $\S$ and computable functions, the class of $\C$'s is Turing equivalent), and give $s$ as input to $\A$. If $\A$ outputs a Turing Machine that, on input $p$ (for a nontrivial $\F_p$), outputs `no solution' or $>0$ interventions, then $M$  outputs $0$, else (i.e., for $0$ interventions) it outputs $1$. Then $\M$ decides whether an input Turing machine has language in $\F_p$, contradicting Theorem~\ref{thm:rice}. 
\end{proof}



\section{Some examples of reverse engineered systems}\label{sec:examples}

\begin{figure}[ht]
\begin{center}
\centerline{\includegraphics[width=\columnwidth]{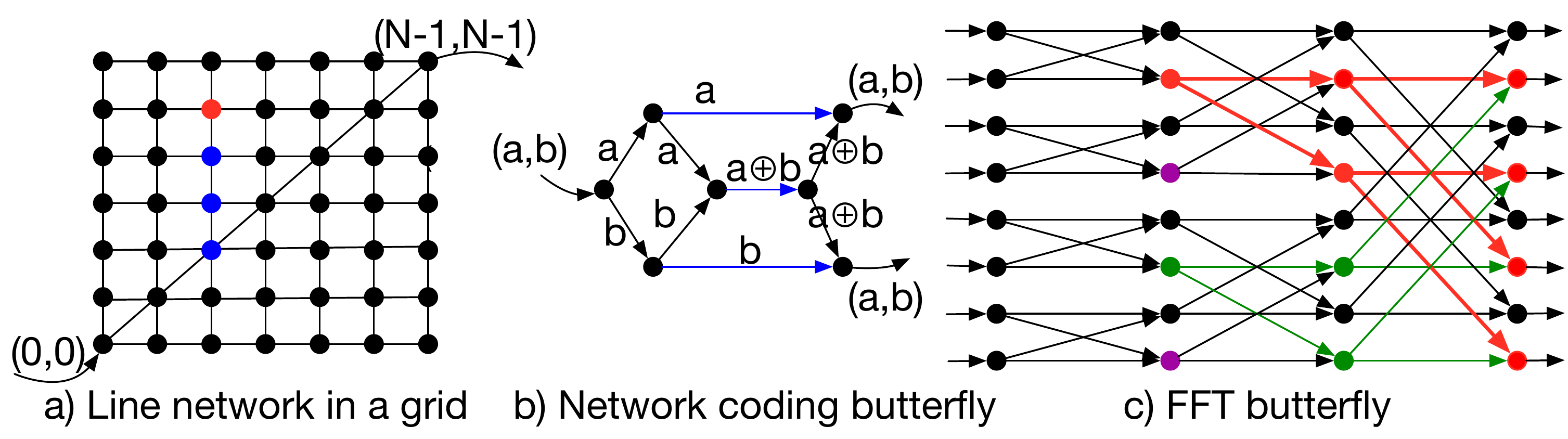}}
\caption{Examples of the first three $\C$'s considered in Section~\ref{sec:examples} for reverse engineering. Input nodes have incoming arrows (with no source), output nodes have outgoing arrows (with no destination). In a), an alternative destination node is shown in red, and blue nodes show where interventions need to be performed to change the I/O relationship to have the output go out of this red node. In c), an example set of output nodes that have their I/O relationship changed are shown. Also shown are pathways which could be affected to cause changes in their behavior.}
\label{fig:examples}
\end{center}
\end{figure}

\begin{example}[Line communication network]
Here, $\C$ is an $N^2$-node network arranged as an  $N\times N$ grid and connected using bidirectional links in the pattern shown in Fig.~\ref{fig:examples}a. The path along a diagonal, going from the (0,0)-node to the (N-1,N-1)-node, is a communication path, with inputs coming to the (0,0)-node, and traversing this path to leave the (N-1,N-1)-node. The set $\H$  contains all sets of I/O relationships, denoted by $\F_{i,j}$, where the $(i,j)$-th node is the destination of communication (i.e., the output of the $(i,j)$-th node is the communication message).
\end{example}
\textit{Reverse engineering $\C$}: $\A$ declares that it has $(\H,L,O(\log(N)),0)$ RE'ed this network for any $L\geq 0$. To do so, $\A$ first identifies the lone information path in the system, namely, the diagonal. The TM output by $\A$ receives as input $(i,j)$, and simply outputs a set of $|i-j|$ nodes that connect $(i,j)$ to the diagonal (namely, if $i>j$, then  $\{(i,j),(i-1,j),\ldots,(j,j)\}$, and symmetrically if $j>i$; note that this is one among many minimal paths to the diagonal from the $(i,j)$-th node).  If the number of nodes in this path exceeds $L$, the TM can declare ``no solution.'' This algorithm requires the TM to store (i) the indices of the node coordinates $(i,j)$ (requiring $O(\log N)$ bits of memory), and (ii) instructions for execute this simple algorithm of reducing one of the two indices (whichever is larger) until they are both equal (requiring constant memory).


\begin{example}[Network-coding butterfly]
Here, $\C$ is the network-coding butterfly network from Ahlswede et al.'s network coding work~\cite{Ahlswede2000Network}. Briefly, two binary symbols, $a$ and $b$, are communicated at both outputs, despite rate limitation on all links of 1 bit, by utilizing an XOR operation in the middle link (see Fig.~\ref{fig:examples}b). $\F_p$ is the set of all changed I/O relationships (not equal to the original butterfly network) where  a) only the first output node is affected (indexed by $p=1$); b) only the second output node is affected ($p=2$); c) both output nodes are affected ($p=3$).
\end{example}
\textit{Reverse Engineering $\C$}: $\A$ declares that it has $(\H,1,M,0)$ RE'ed $\C$ (with $M$ specified below). For the network-coding butterfly, a single intervention suffices. For $p=1$, an intervention on the top edge, for $p=2$, an intervention on the bottom edge, and for $p=3$, an intervention on the middle edge suffice. $M$ is simply the length of a (e.g., the smallest) Turing machine that outputs the correct intervention for the input $p$.


\begin{example}[$N$-point FFT-butterfly network] Here, $\C$ is the FFT butterfly network for computing the $N$-point FFT on a finite field~\cite{pollard1971fast}. $\H$ is the collection of I/O relationship sets $\F_p$ where any single $\F_p$ is the set of all possible changed I/O relationships that only affect a fixed subset of the output nodes (the subset is indexed by $p$) in the butterfly network. 
\end{example}
\textit{Reverse engineering $\C$: } $\A$ declares that it has ($\H,L=1,M=O(N),Q=1)$) RE'ed $\C$. I.e., $\A$ declares that it can a) label which I/O relationship sets are obtained by intervening on a single node; b) Output a single node that on intervention yields the desired I/O characteristics; and c) use $M=O(N)$ memory and one multi-intervention query ($Q=1$) in doing so.  The key observation is that (see Figure~\ref{fig:examples}c) if an I/O change inside a $\F_p$ can arise from interventions on a single node, then one such node is the one that we arrive at by stepping leftwards (by $\log_2(B)$ steps if $B$, the number of affected output nodes, is $2^k$ for some $k\in\mathbb{N}$) from \textit{any} of the affected output nodes (see Fig.~\ref{fig:examples}c for intuition).

The TM output by $\A$ executes the following: the input $p$ provides the indices of the output nodes affected by the intervention. If the number of these nodes is not $2^k$ for some $k\in\mathbb{N}$, output ``no solution'' (a single intervention is insufficient). If it is, then choose the first such output node, and, looking at the FFT architecture, traverse left by $k$ steps. Ask the oracle if an intervention on this node can produce a desired I/O pattern. If yes, then a solution is this node. If not, output ``no solution'' ($>1$ interventions needed).

\begin{figure}[ht]
\begin{center}
\centerline{\includegraphics[width=0.3\columnwidth]{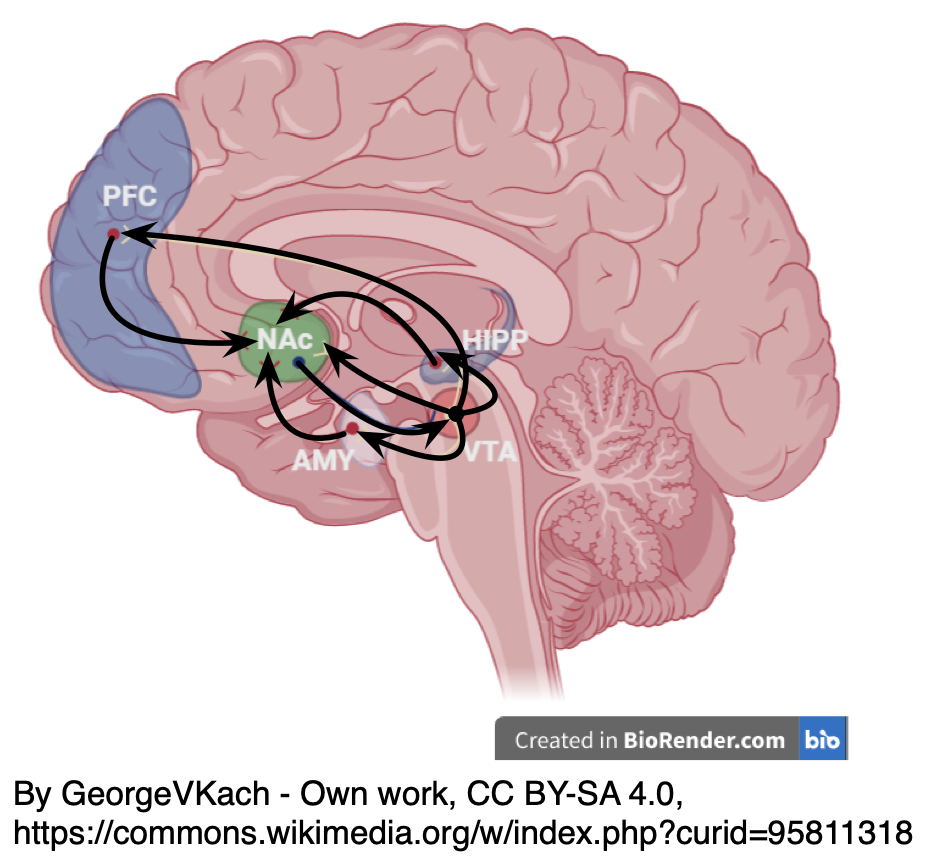}}
\caption{A simplified reward network in the brain for humans (edited to clearly illustrate directions of links. Original downloaded from Wikipedia. Usage license: By GeorgeVKach -  CC BY-SA 4.0, 
\url{https://commons.wikimedia.org/w/index.php?curid=95811318)}. A more detailed figure is in~\cite{russo2013brain}, which furhter illustrates some of the salient nodes and links, including VTA: Ventral Tegmental Area, AMY: Amygdala, HIPP: Hippocampus, PFC: Pre-Frontal Cortex, NAc: Nucleus Accumbens. }
\label{fig:reward}
\end{center}
\end{figure}

\section{\PGedit{Examining the state of understanding of an exemplar brain network: the reward network in the brain}}
\label{sec:reward}

The brain's reward network is a complex circuit that is responsible for desire for a reward, positive reinforcement, arousal, etc. Dysfunction in this network can result in depression, obsessive-compulsive disorder (OCD), addiction, etc. The reward network consists of several large centers, such as the ventral tegmental area (VTA), the Amygdala (Amy), the Nucleus Accumbens (NAc), the Hippocampus (Hipp), the Prefrontal Cortex  (PFC), the Orbitofrontal Cortex (OFC), etc.,  that interact with one another in complex ways. A simplified version of this network is illustrated in  Fig.~\ref{fig:reward}. 

\PGedit{Decades of scientific research has helped develop some understanding of how these large brain regions interact. Below, we provide a brief overview of this body of work in the context of representation of ``valence'' (positive or negative emotion) in the reward network. We refer biologically-inclined readers to~\cite{namburi2016architectural,russo2013brain} as starting points for a deeper study. This overview summarizes the understanding of the reward network as it stands today, and how it can suggest strategies for interventions. We want the reader to observe that, while the understanding is quite detailed, it is still far from that needed for the reverse-engineering goal laid out in our work. This discrepancy could help set an aim for neuroscientists, but also help expand (in subsequent work) our computer-scientific definitions to include limitations of the understanding of, and/or the ability to intervene on, this circuit (e.g. if some nodes are inaccessible for stimulation, or less explored for their functional understanding).}

\PGedit{Back in 1939, Kl\"{u}ver and Bucy~\cite{kluver1939preliminary} observed (in monkeys) that lesioning in the temporal lobe and amygdala led to extreme emotional changes, including loss of fear responses, failure to learn from aversive stimuli, and increased sexual behavior (leading to what is called Kl\"{u}ver-Bucy syndrome in humans with similar injuries). Since this work, animal studies, including in mice, rats, monkeys, etc., have  been frequently used to understand how the brain responds to rewarding/pleasant (positively valenced) or aversive (negatively valenced) stimulus presentation.  Many studies have since examined which regions of the brain ``represent valence'', in that their neural response statistics change when positive vs negatively valenced stimuli are presented. These studies show that many (broad) regions represent valence, including the amygdala~\cite{fuster1971reactivity,shabel2009substantial}, nucleus accumbens~\cite{roitman2005nucleus}, ventral tegmental area~\cite{matsumoto2009two}, orbitofrontal and prefrontal cortex~\cite{schoenbaum1999neural}, lateral hypothalamus~\cite{fukuda1990dopamine}, subthalamic nucleus~\cite{sieger2015distinct}, hippocampus~\cite{fuster1971reactivity}, etc. (see~\cite{namburi2016architectural} for an excellent survey). Recently, advances in neuroengineering, especially in optogenetics~\cite{boyden2005millisecond} and minimally invasive implants~\cite{neuraldust}, enable finer-grained examination \textit{within} these broad  brain regions, including spatiotemporally precise interventions, examining neural ``populations'', i.e., collections of neurons within the same broad region that are similar ``functionally'' (i.e., in how they respond to rewarding or aversive stimuli), genetically (e.g., in the type of neurons), and/or in their connectivity (which region they connect with). For Nucleus Accumbens, for instance, these techniques have led to further separation of the region into its \textit{core} vs its \textit{shell}. Dopamine release in the core (often due to activation of the VTA by a rewarding or aversive stimulus) appears to reinforce rewarding behavior, while same dopamine released in the shell can lead to both rewarding and aversive stimuli. E.g. an addiction `hotspot' is found in the medial shell, while in another location, a `coldspot' reduces response to addictive stimuli, suggesting a fine control by the two populations (see~\cite{namburi2016architectural}). Similarly refined understanding has been developed for other nodes, e.g. the amygdala and the VTA (see~\cite{namburi2016architectural}).
}

\PGedit{Thus, at first glance, one way think that that estimation of what stimulus presentation affects which neural population, and how interventions on a neural population affect processing of a stimulus, are increasingly at a spatial resolution that is required to answer reverse engineering questions we pose here (they will only be further enabled by recent advances in neuroengineering~\cite{boyden2005millisecond,neuraldust}). \textit{Indeed, many clinicians are already utilizing this understanding to do surgical implants} that intervene on functioning of this network,  including for depression~\cite{bewernick2010nucleus}, OCD~\cite{blomstedt2013deep}, addiction~\cite{pierce2013deep}, obesity~\cite{oh2009deep}, etc., when the disorder is extreme. However, our understanding of the network is still severely lacking: we do not know, for instance, what the functions computed at these nodes are, which can have a significant effect on what the minimal intervention is. }

These limitations in understanding of this network affects our ability to provide optimized solutions (e.g. those that are minimal in the sense discussed in our paper). This might seem intuitive, but for completeness we include a simple example  \PGedit{of the influence of the Nucleus Accumbens on subsequent nodes (PFC and OFC). E.g., suppose its output to PFC, $Y_{NAc\rightarrow PFC}=I_{HS}- I_{CS}$ is the difference of the outputs of the hotspot $I_{HS}$ and the coldspot $I_{CS}$ discussed above.  Further, the output to OFC could be A) the ratio; or B) the difference of the outputs of the hotspot and the coldspot. That is, $Y_{NAc\rightarrow OFC}^{(1)}=I_{HS}/I_{CS}$ and $Y_{NAc\rightarrow OFC}^{(2)}=I_{HS}- I_{CS}$.  The goal is to produce an intervention that makes $Y_{NAc\rightarrow OFC}=Y_{NAc\rightarrow PFC}=(1-\gamma) I_{HS}$ (i.e., $\H$ is constituted by the I/O relationships of this form for NAc, one for each $\gamma$). Now, let's assume that links (arising from separate nodes) from the hotspot and the coldspot populations go to PFC and OFC, but the coldspot receives $I_{CS}$ from a common ancestor. It is easy to see that in this case, the reverse engineering solution depends on which is the actual function: if $Y_{NAc\rightarrow PFC}=Y_{NAc\rightarrow OFC}^{(2)}=I_{HS}- I_{CS}$, intervening on the coldspot's ancestor will suffice (namely, by setting $I_{CS}= \gamma I_{HS}$). However, if $Y_{NAc\rightarrow OFC}$ is the ratio,   $Y_{NAc\rightarrow OFC}^{(1)}=I_{HS}/I_{CS}$, the set of minimal interventions could be of cardinality two, constituted by interventions on two locations within the coldspot, to get both outputs to equal $(1-\gamma) I_{HS}$ (namely, one that outputs to PFC should have the signal $\gamma I_{HS}$, whereas one that outputs to OFC should have the signal $1/(1-\gamma)$). Observe that the qualitative relationship between how $I_{HS}$ and $I_{CS}$ affect the outputs is similar in the two possibilities considered here (i.e., the first increases the outputs, and the second reduces it).  }

We think that this suggests the possibility of subsequent work which uses a computer-scientific and information-theoretic lens to contribute to design of experiments (observational and interventional) for garnering the needed inferences about this computational system (such as modeling functions computed at nodes, not just activation/influence of a node).

\section{Discussion and limitations}\label{sec:discussion}



What aspect in our work makes it motivated by neuroscience? After all, our computation system model is fairly general, and builds on prior work in theoretical computer science (see, for instance, work on ``VLSI theory'' in the 1980s~\cite{thompson1979area,Thompson1980Complexity}, which motivated models in~\cite{venkatesh2020information,venkatesh2019should,venkatesh2020else} that we are, in turn, inspired by). While,  intellectually, finding a set of minimal interventions demonstrates strong understanding of how a computational system works, we believe that, operationally, the minimal intervention aspect is most closely tied to networks in neuroscience. Intervening on machine-learning networks (such as ANNs), we can find no natural reason why one should attempt to find minimal interventions. Editing few nodes and/or edges of ANNs implemented in hardware is not a problem that is relevant in today's implementations. However, this problem arises naturally in neuroscience, as one would naturally want to to intervene on as few locations as possible (say, because each location requires a surgical intervention). Of relevance here is a recent work on cutting links in epileptic networks, where the authors seek a similar minimalism~\cite{kini2019virtual}.

{\textit{Our definition of what constitutes a \underline{minimal} intervention} could be tied more closely to biological constraints and peculiarities. While our definition is  motivated from recent surgical interventions on the reward circuitry and advances in neuroscience, sometimes, a noninvasive intervention, even if more diffused, might be preferred to an invasive intervention because it does not require implantation (implantation has risk of infection, need for removal etc.). Similarly, it is known that in the brain (even in the reward network~\cite{namburi2016architectural}), different populations have different likelihood of having neurons that represent and affect valence, and different neurons also have different magnitudes of effects they produce on the network's reward valuation. The practical difficulties of finding a neuron close to where an implant is placed, and/or difficulty-levels of surgical interventions, might need to be incorporated in our model. }

As a practical direction, we think that clinical neuroscience research should not only focus on describing the system or examining some causal pathways of interventions, but also actively on modifications and interventions at the fewest possible locations (or minimal in ways suited to the specific disorder) that can change the I/O behavior to a desirable one. It is conceivable that a neuroscientist might want to demonstrate how they are able to ``control'' the circuit as a way of certifying their understanding of the system. From this perspective, we recognize that this demonstration of control (to any I/O behavior) of the circuit is stronger than what might be needed for getting a \textit{specific} behavior that is desirable, and this can be captured in our definition by careful choice of $\H$.


Our nodes-and-edges discrete-time model is a crude one, because even single cells can exhibit extremely complex dynamics~\cite{hodgkin1952quantitative,izhikevich2007dynamical}. However, models such as ours are commonly used (e.g.~\cite{venkatesh2020information,Bressler2011Wiener,Brovelli2004Beta} and  references therein) in computational neuroscience as a first step, and have been applied to real data. Here, our goal is to use these models to formally state the reverse engineering definition, which allows us to illustrate how reverse engineering could be achieved, and obtain undecidability results for a class of problems.

On our requirements, one can replace bounded memory constraints to other constraints~\cite{simon1990bounded} (e.g., computational or informational~\cite{sims2005rational}),  or also seek \textit{approximately} minimal interventions. We believe that (based on simplicity of results in Appendix~\ref{sec:alternative_definitions}) the general problem will continue to be undecidable for many such variations. Hardness/impossibility results have continued to inform and refine approaches across many fields (e.g.~hardness of Nash equilibrium~\cite{daskalakis2009complexity} and of decentralized control problems~\cite{papadimitriou1986intractable}, and recent undecidability results in physics ~\cite{cubitt2015undecidability,kreinovich2017some}, among others). An undeniable consequence of our result is that there cannot exist an algorithm that solves the reverse engineering problem posed here \textit{in general}. There exist cases that are extremely hard to reverse engineer, even if (as illustrated by our examples in Section~\ref{sec:examples}), in many cases, reverse engineering \textit{can} be accomplished. 

On our undecidability result, note that if the alphabet of computation is finite, then the reverse engineering problems posed here are decidable. However, in that case, the model for brain is also not Turing complete. Finally, one must note that undecidability is not an artifact of our definition. {As shown in Appendix \ref{sec:alternative_definitions}, other plausible definitions we considered also yielded analogous undecidability results.} Our proof technique extends to many related definitions, as illustrated by the relaxed assumptions under which we are able to prove the  results (and, indeed, the relaxed assumptions under which Rice's theorem is obtained). As rapid advances in neuroengineering enable breakthrough neuroscience, challenging conceptual and mathematical problems will arise.  In fact, today, both AI and neuroscience are using increasingly complex models and are asking increasingly complex interpretability/reverse engineering questions. It is worth asking whether instances of this question are undecidable, and, if decidable, how the complexity of a reverse engineering problem scales with the problem size. 





\newpage{}




\bibliography{ITCS_Submission}


\appendix

\section{Proof of Theorem~\ref{thm:equivalence}.1}

We denote a Deterministic Finite Automaton by DFA $(\Sigma, Q_{dfa}, s_0, \delta, F)$ consisting of a finite set of states $Q_{dfa}$, a finite set of input symbols called the alphabet $\Sigma$, a transition function $\delta :Q_{dfa}\times \Sigma \rightarrow Q_{dfa}$, 
an initial state $q_{0}\in Q_{dfa}$, and a set of accepting states $F\subseteq Q_{dfa}$.  To simulate the DFA $(\Sigma, Q_{dfa}, s_0, \delta, F)$, we construct a $\C$ as follows:
	\begin{enumerate}
		\item Nodes of the graph $\G$ in $\C$ are the states of the DFA, with an additional output node $\{o\}$, i.e., $\V = Q_{dfa}\cup\{o\}$. 
		\item Edges of $\G$ are: (i) all the transition edges of the DFA, i.e. for every two states $s,t\in Q_{dfa}$ for which there is $d\in \Sigma$ such that $\delta(s, d)=t$, there is an edge $s\to t$ in $\G$; (ii) self-loops at every node (if not defined by (i)); and (iii) For each accepting state of the DFA ($s\in F$), an edge $s\to o$.
		\item $\S=\Sigma \cup \{ \mathrm{start} , \mathrm{fin} , \mathrm{blank}\}$, i.e., for defining $\S$ in $\C$, we use $\Sigma$, and,  additionally,  $\mathrm{start}$,  $\mathrm{fin}$ (finish), and $\mathrm{blank}$ symbols.
		\item Each node receives the computational input, $d$, at each time step. 
		\item Initialize all states of nodes of $\C$, except the node corresponding to $s_0$, with $\mathrm{blank}$, and the state of the node corresponding to $s_0$ with $\mathrm{start}$. The function computed at each node $s\in \V\setminus \{o\}$, 
		on the transmissions it receives (say $x_1, \dots, x_k$; exactly one of the $x_j$'s is not $\mathrm{blank}$) and the computation input $d$ is
		\begin{align*}
			&f_s(\blank ,.., x_j ,..,\blank, d)\\&\hspace{0.3in}=
			\begin{cases}
				d, &\text{if }x_j \in \Sigma, \delta(s_j, x_j)=s\\
				d, &\text{if }x_j = \text{`}\mathrm{start}\text{'}, s=s_0\\
				\blank &\text{else},
			\end{cases}
		\end{align*}
        and the output node computes:
        \begin{equation*}
            f_o(\blank,\dots, x_j=`\mathrm{fin}\text{'}, \dots,\blank, d) = 
            \begin{cases}
                1, \text{if } s_j\in F\\
                0, \text{else}.
            \end{cases}
        \end{equation*}
	\end{enumerate}
	With this construction, the output node outputs $1$ on a computational input string iff the DFA accepts the string.

\section{Alternative Definitions}\label{sec:alternative_definitions}

Here, we introduce two \textit{non-interventional} definitions of reverse engineering, and show that those are also undecidable.

\begin{definition}[Single-Node RE]
    An agent $\A$ is said to Single-Node Reverse Engineer a computational system $\C$ if given any node $i$ of $\C$, it can determine whether there is any input to the computation system such that at some time instant, the node $i$ stores a non-zero value (i.e. whether the node is ever activated).
\end{definition}

\begin{theorem}
    There is no Turing machine $\A$ which can accept as input, an arbitrary computational system $\C$ (having countably infinite alphabet) and arbitrary node $i$ of $\C$, and output whether the node $i$ is ever activated.
\end{theorem}
\begin{proof}
    Suppose there were such a Turing machine $\A$. Then, we can construct a Turing machine $\M_E$ that decides the language 
    \[
        E_{TM} = \{ \langle M \rangle : M \text{ is a TM and } L(M) = \emptyset \}
    \]
    as follows: $\M_E$ accepts input string encoding Turing machine $\langle M \rangle$, creates an encoding of the corresponding computation system $\C_M$ whose output node is labelled as $v$. $\M_E$ simulates $\A$ on input $\langle \C_M, v \rangle$ and outputs true iff $\A$ outputs true. 
    
    Then $\M_E$ as described above decides $E_{TM}$ since node $v$ of the constructed computation system $\C_M$ is ever activated iff $M$ ever accepts an input string. However we know that $E_{TM}$  is undecidable (Theorem 5.2, \cite{sipser13}), thus such a Turing machine $\M_E$ cannot exist. 
\end{proof}

This previous result shows that determining if a node in a neural circuit even represents a message of interest (e.g. positive or negative valence of a reward in Section~\ref{sec:reward}) is undecidable. The result that follows this next definition shows that even estimating approximations of functions being computed (I/O relationships) can be undecidable. 

\begin{definition}[$(k,f)$-Approximate RE]
    Given a computable function $f:\Q\to\Q$ and number $k\in \N$, an agent $\A$ is said to $(k,f)$-Reverse Engineer a computational system $\C$, if it can determine whether $\C$ computes a $k$-approximation of $f$, i.e., whether on every input string $x\in \Q$, we have
    \[
        \frac{1}{k} |f(x)| \leq |\C(x)| \leq k |f(x)|
    \]
\end{definition}

\begin{theorem}
    For every computable function $f$, 
    there is no Turing machine $\A$ which can accept as input an arbitrary computation system $\C$ and output whether $\C$ computes a $k$-approximation of $f$. 
\end{theorem}
\begin{proof}
    As in the previous theorem, suppose there were such a Turing machine $\A$. Then, we construct a Turing machine $\M_E$ deciding $E_{TM}$ as follows: $\M_E$ accepts input string encoding Turing machine $\langle M \rangle$. It constructs an encoding $\langle \C_M\rangle$ of a computation system which takes input string $x$, first simulates computing $M(x)$. Then, if $M$ accepts $x$, outputs $k f(x) + 1$, else outputs $f(x)$. Then $\M_E$ simulates $\A$ on input $\langle \C_M\rangle$ and outputs true iff $\A$ determines that $\C$ is a $k$-approximation of $f$. 
    
    Thus, $\M_E$ described as above decides $E_{TM}$ since the constructed $\C_M$ computes a $k$-approximation of $f$ iff $M$  rejects all inputs. However, as we know, $E_{TM}$ is undecidable. Thus by contradiction, such an $\A$ does not exist. 
\end{proof}

\end{document}

%% file: mathdefs.tex
\newcommand{\xor}{\oplus}